\documentclass[11pt]{article}

\usepackage{graphicx}

\usepackage{amsmath}
\usepackage{amssymb}
\usepackage{amsthm}
\usepackage{color}

\newtheorem{theorem}{Theorem}
\newtheorem{proposition}{Proposition}

\begin{document}
\date{ }

\title{An allocation scheme for estimating the reliability of a
parallel-series system}

\author{Z. BENKAMRA$\ ^{1}$, M. TERBECHE$\ ^{2}$ \ and M. TLEMCANI$\ ^{1,*}$ \\
\\$\ ^{1}$ University Mohamed Boudiaf, L.A.A.R, Algeria\\
$\ ^{2}$ University of Oran, Algeria\\
\\$\ ^{*}$ mounir.tlemcani@univ-pau.fr (M.Tlemcani)}%

\maketitle

\begin{abstract}
We give a hybrid two stage design which can be useful to estimate the
reliability of a parallel--series and/or by duality a series--parallel
system, when the component reliabilities are unknown as well as the total
numbers of units allowed to be tested in each subsystem. When a total sample
size is fixed large, asymptotic optimality is proved systematically and
validated \textit{via} Monte Carlo simulation.
\end{abstract}



\section{Introduction}
In reliability engineering two crucial objectives are considered: (1) to
maximize an estimate of system reliability and (2) to minimize the variance
of the reliability estimate. Because system designers and users are
risk-averse, they generally prefer the second objective which leads to a
system design with a slightly lower reliability estimate but a lower
variance of that estimate {, (eg, \cite{coit}). It provides decision makers efficient rules compared to other
designs which have a higher system reliability estimate, but with a high variability of that estimate.}
In the case of parallel--series and/or by duality series--parallel systems, the variance of
the reliability estimate can be lowered by allocation of a fixed sample
size ({the number of observations or units tested in the system}), while reliability estimate
is obtained by testing components, {see Berry \cite{berry}}.
Allocation schemes for estimation with cost, see for example \cite{berry,djerjour,hardwick,page,terbeche,woodroofe},
lead generally to a discrete optimization problem which can be solved
sequentially using adaptive designs in a fixed or a {Bayesian} framework. {Based on a decision theoretic approach,
the authors seek to minimize either the variance or the Bayes risk associated to a squared error loss function.
The problem of optimal reliability estimation reduces to a problem of optimal allocation of the sample sizes
between Bernoulli populations. Such problems can be solved \textit{via} dynamic programming but this technique becomes
costly and intractable for complex systems. In the case of a two components series or parallel system, optimal procedures
can be obtained and solved analytically when the coefficients of variation of the associated Bernoulli populations are known,
cf. eg, \cite{hardwick-bp,page}. Unfortunately, the coefficients of variation are not known in practice since they depend themselves on
the unknown components reliabilities of the system.
In \cite{rekab}, the author has defined a sequential allocation scheme in
the case of a series system and has shown its {first} order asymptotic optimality for large sample sizes
with comparison to the balanced scheme}.
In \cite{arima}, a reliability sequential schemes (R-SS) was applied successfully to parallel--series
systems, when the total number of units to be tested in each subsystem was
fixed. Recently, in \cite{matcom}, a two stage design for the same purpose was
presented and shown to be asymptotically optimal when the subsystems sample
sizes are fixed and large at the same order of the total sample size of the
{system}.
The problem considered in this paper is useful to estimate the
reliability of a parallel-series and/or by duality a series-parallel system,
when the components reliabilities are unknown as well as the total numbers of
units allowed to be tested in each subsystem. This work improves the results
in \cite{matcom} by developing a hybrid two stage design to get a dynamic
allocation between the sample sizes allowed for subsystems and those allowed
for their components. 
{For example, consider a parallel system of four components (1),(2),(3) and (4), with reliabilities 0.05, 0.1, 0.95 and 0.99, respectively,
under the constraint that the total number of observations allowed is $T=100$. Then, the sequential scheme given in \cite{arima}
suggests to test, respectively, 10, 10, 28 and 52 units and produces a variance of the system reliability estimate equal $10^{-7}$,
approximately. This is visibly better, compared to the balanced scheme which takes an allocation equal 25 in each component and produces
a variance ten times greater then the former. The hybrid sequential scheme proposed in this paper is a tool to solve the same problem when
the components are replaced by subsystems. More precisely, it combines the schemes developed for parallel and/or series systems in order
to obtain approximately the best allocation at subsystems level as well as at components level.}

In section \ref{prelim}, definitions and preliminary results are presented
accompanied by the proper two stage design for a parallel subsystem just as
was defined in \cite{matcom} and its asymptotic optimality is proved for a
fixed and large sample size. In section \ref{lower}, {a parallel-series}
system is considered and it is shown that the variance of
its reliability estimate has a lower bound independent of allocation. This
leads, in section \ref{overlapping}, to the main result of this paper
which lies in the hybrid two stage algorithm and its asymptotic optimality
for a fixed and large sample size allowed for the {system}. In section %
\ref{monte carlo}, the results are validated \textit{via} Monte Carlo
simulation and it is shown that our algorithm leads asymptotically to the
best allocation scheme to reach the lower bound of the variance of the
reliability estimate. The last section is reserved for conclusion and
remarks.
\section{Preliminary results}
\label{prelim}
Consider a system $S$ of $n$ subsystems $S_{1},S_{2},\ldots ,S_{n}$
connected in series, each subsystem $S_{j}$ contains $n_{j}$ components $%
S_{1j},S_{2j},\ldots ,S_{n_{j}j}$ connected in parallel. The system should
be referred as parallel-series system. Assume s-independence
within and across populations, then the system reliability is 
\begin{equation}
\label{r}
R=\prod\limits_{j=1}^{n}R_{j},
\end{equation}
where 
\[
R_{j}=1-\prod\limits_{i=1}^{n_{j}}\left( 1-R_{ij}\right)
\]%
is the reliability of the parallel subsystem $S_{j}$ and $R_{ij}$ the
reliability of component $S_{ij}$. An estimator of $R$ is assumed to be the
product of sample reliabilities%
\[
\hat{R}=\prod\limits_{j=1}^{n}\hat{R}_{j},
\]
where 
\[
\hat{R}_{j}=1-\prod\limits_{i=1}^{n_{j}}\left( 1-\hat{R}_{ij}\right) 
\]
and $\hat{R}_{ij}$ is the sample mean of functioning units in component $%
S_{ij}$,%
\[
\hat{R}_{ij}=\frac{\sum\limits_{l=1}^{M_{ij}}X_{ij}^{(l)}}{M_{ij}},
\]%
$\hat{R}_{ij}$ is used to estimate $R_{ij}$ where $M_{ij}$ is the sample
size and $X_{ij}^{(l)}$ is the binary outcome of the unit $l$ in component $%
S_{ij} $. {It should be pointed that a unit 
is not necessarily a physical object in a component, but it represents just a Bernoulli observation of
the functioning/failure state of that component.} 
Hence, for each subsystem $S_{j}$, one must allocate%
\[
T_{j}=\sum\limits_{i=1}^{n_{j}}M_{ij} 
\]%
units such that the estimated reliability of the {system} is based on a
total sample size%
\[
T=\sum\limits_{j=1}^{n}T_{j}
\]%
As in the series case, {with the help of s-independence
and the fact that a sample mean is an unbiased estimator of a Bernoulli parameter,
see \cite{arima,matcom,coit}}, the variance of the estimated reliability $\hat{R}$
incurred by any allocation scheme can be obtained, 
\begin{equation}
Var\left\{ \hat{R}\right\} =\prod\limits_{j=1}^{n}\left( Var\left( \hat{R}%
_{j}\right) +R_{j}^{2}\right) -\prod\limits_{j=1}^{n}R_{j}^{2},  \label{g}
\end{equation}%
where 
\begin{equation}
Var\left\{ \hat{R}_{j}\right\} =\left( 1-R_{j}\right) ^{2}\left[
\prod\limits_{i=1}^{n_{j}}\left( 1+\frac{c_{ij}^{-2}}{M_{ij}}\right) -1%
\right]  \label{Vp}
\end{equation}%
is given as a function of the allocation numbers $M_{ij}$ and the
coefficients of variation of Bernoulli populations%
\[
c_{ij}=\sqrt{1/R_{ij}-1} 
\]%
We have found convenient to work with the equivalent expression of (\ref{Vp}),
\[
\hspace{-0.15cm} Var\left\{ \hat{R}_{j}\right\} =\left( 1-R_{j}\right) ^{2}%
\left[ \sum\limits_{i=1}^{n_{j}}\frac{c_{ij}^{-2}}{M_{ij}}+F\left( \frac{%
c_{1j}^{-2}}{M_{1j}},...,\frac{c_{n_{j}j}^{-2}}{M_{n_{j}j}}\right) \right], 
\]
where%
\[
F\left( \frac{c_{1j}^{-2}}{M_{1j}},...,\frac{c_{n_{j}j}^{-2}}{M_{n_{j}j}}%
\right) 
\]%
is a sum over all the products of at least two of its arguments.

The problem is to estimate $R$ when components reliabilities are unknowns
and a total number of $T$ units must be tested in {the} system {at components level}. The aim is 
to minimize the variance of $\hat{R}$. Hence, the problem can be addressed by
developing allocation schemes to select $M_{ij}$, the numbers of units to be
tested in each component $i$ in the subsystem $j$, under the constraint%
\begin{equation}
\sum_{j=1}^{n}\sum_{i=1}^{n_{j}}M_{ij}=T,  \label{constr}
\end{equation}%
such that the variance of $\hat{R}$ is as small as possible. Reliability
sequential schemes (R-SS) exist for the series, parallel or parallel-series
configurations when the sample sizes $T_{j}$ of the subsystems are fixed.
Therefore, one can fully optimize the variance of $\hat{R}$ just by applying
the (R-SS) to find the best partition $T_{1},T_{2},...,T_{n}$ of $T$.
Unfortunately, a full sequential design can not be used in practice for large systems since
the number of operations will growth dramatically. For this reason, we reasonably propose a
hybrid two stage design which is shown to be asymptotically optimal when $T$ is large.
\subsection{Lower bound for the variance of the estimated reliability of the
parallel subsystem $S_{j}$\label{par}}
For the asymptotic optimization of the variance of the estimated reliabilities, we make use of the well-known
Lagrange's identity which can be written in the form:

Let $a_{i}>0$, $N_{i}>0,$ for $i=1,...,k$ and $N=N_{1}+\cdots
+N_{k}$, then the following identity holds.%
\begin{equation}
\label{lem1}
\sum\limits_{i=1}^{k}\frac{a_{i}}{N_{i}} = N^{-1}\left[\left( \sum\limits_{i=1}^{k}\sqrt{a_{i}}%
\right) ^{2}+\sum\limits_{i=1}^{k-1}\sum\limits_{j=i+1}^{k}\frac{\left( N_{i}\sqrt{%
a_{j}}-N_{j}\sqrt{a_{i}}\right) ^{2}}{N_{i}N_{j}}\right]
\end{equation}

\begin{proposition}
\label{th0}Denote by 
\begin{equation}
Q_{j}=\left( 1-R_{j}\right) ^{2}T_{j}^{-1}\left(
\sum_{i=1}^{n_{j}}c_{ij}^{-1}\right) ^{2}  \label{qj}
\end{equation}%
then%
\[
Var\left\{ \hat{R}_{j}\right\} \geq Q_{j}
\]
\end{proposition}

\begin{proof}
The proof is a direct consequence of the previous identity (\ref{lem1}). Indeed%
\begin{eqnarray}
&&Var\left\{ \hat{R}_{j}\right\} =\left( 1-R_{j}\right) ^{2}T_{j}^{-1}\left(
\sum_{i=1}^{n_{j}}c_{ij}^{-1}\right) ^{2}\nonumber    \\
&&+T_{j}^{-1}\left( 1-R_{j}\right)
^{2}\sum_{i=1}^{n_{j}-1}\sum_{k=i+1}^{n_{j}}\frac{\left(
M_{ij}c_{kj}^{-1}-M_{kj}c_{ij}^{-1}\right) ^{2}}{M_{ij}M_{kj}}\nonumber    \\
&&+\left( 1-R_{j}\right) ^{2}F\left( \frac{c_{1j}^{-2}}{M_{1j}},\frac{%
c_{2j}^{-2}}{M_{2j}},...,\frac{c_{n_{j}j}^{-2}}{M_{n_{j}j}}\right)
\label{var rj}
\end{eqnarray}
\end{proof}
\subsection{The two stage design for the parallel subsystem $S_{j}$}
Following the expansion (\ref{var rj}) and since $F$ contains second order terms (see later), one gives interest to the numbers $%
M_{ij}$ which minimize the expression%
\[
T_{j}^{-1}\sum_{i=1}^{n_{j}-1}\sum_{k=i+1}^{n_{j}}\frac{\left(
M_{ij}c_{kj}^{-1}-M_{kj}c_{ij}^{-1}\right) ^{2}}{M_{ij}M_{kj}} 
\]%
Thus $M_{ij}$ must verify for $i=1,...,n_{j}$%
\[
M_{ij}c_{kj}^{-1}=M_{kj}c_{ij}^{-1} 
\]%
which implies that%
\begin{equation}
\label{mij}
M_{ij}=T_{j}\frac{c_{ij}^{-1}}{\sum\limits_{k=1}^{n_{j}}c_{kj}^{-1}}
\end{equation}
If one assumes that $T_{j}$ is fixed then a proper two stage scheme can be
used to determine $M_{ij}$, just as was defined in \cite{matcom}, as follows:

Choose $L_{j}$ as a function of $T_{j}$ such that:

\begin{enumerate}
\item[(i)] $L_{j}$ must be large if $T_{j}$ is large,
\item[(ii)] $L_{j}\leq \frac{T_{j}}{n_{j}},$
\item[(iii)] $\lim\limits_{T_{j}\rightarrow \infty }\frac{L_{j}}{T_{j}}=0$.
\end{enumerate}
One can take for example $L_{j}=\left[\sqrt{T_{j}}\right]$, where $\left[.\right]$ denotes the integer part.
\begin{description}
\item[Stage 1.] Sample $L_{j}$ units from each component $i$ in the subsystem $j$%
, estimate $c_{ij}$ by its maximum likelihood estimator (M.L.E)
\[
\hat{c}_{ij}=\sqrt{\frac{L_{j}}{\sum\limits_{l=1}^{L_{j}}X_{ij}^{(l)}}-1} 
\]
and define the predictor, according to (\ref{mij}),%
\[
\hat{M}_{ij}=\left[T_{j}\frac{\hat{c}_{ij}^{-1}}{\sum\limits_{k=1}^{n_{j}}\hat{c}%
_{kj}^{-1}}\right],~i=1,\ldots ,n_{j}-1 
\]

\item[Stage 2.] Sample $T_{j}-n_{j}L_{j}$ units for which $M_{ij}-L_{j}$ are units
from component $i$ in the subsystem $j$ where $M_{ij}$ is the corrector of $%
\hat{M}_{ij}$ defined by 
\begin{eqnarray*}
M_{ij}&=&\max \left\{ L_{j},\hat{M}%
_{ij}\right\},~ i=1,\ldots,n_{j}-1, \\
M_{n_{j}j}&=&T_{j}-\sum\limits_{k=1}^{n_{j}-1}M_{kj}
\end{eqnarray*}
\end{description}
\begin{theorem}
\label{th1}
Choosing the $M_{ij}$ according to the previous two stage
sampling scheme, one obtains 
\[
\lim_{T_{j}\rightarrow \infty }T_{j}\left( Var\left\{ \hat{R}_{j}\right\}
-Q_{j}\right) =0
\]
\end{theorem}

\begin{proof}
From relation (\ref{var rj}), one can write%
\begin{eqnarray}
&& T_{j}\left( Var\left\{ \hat{R}_{j}\right\} -Q_{j}\right) = \left(
1-R_{j}\right) ^{2}\sum_{i=1}^{n_{j}-1}\sum_{k=i+1}^{n_{j}}\frac{\left(
M_{ij}c_{kj}^{-1}-M_{kj}c_{ij}^{-1}\right) ^{2}}{M_{ij}M_{kj}}\nonumber    \\
&& +\left( 1-R_{j}\right) ^{2}T_{j}.F\left( \frac{c_{1j}^{-2}}{M_{1j}},...,\frac{c_{n_{j}j}^{-2}}{M_{n_{j}j}}\right)
\label{vrj}
\end{eqnarray}%
When $T_{j}$ is large enough, condition (iii) gives $M_{ij}={\hat{M}_{ij}}$ for   $%
i=1,...,n_{j}-1$.
So, the strong law of large numbers with the integer part properties give, when $T_{j}\rightarrow \infty $,%
\[
\frac{M_{ij}}{M_{kj}}{\rightarrow }\frac{%
c_{kj}}{c_{ij}},
\]%
for $i=1,...,n_{j}$. Hence,%
\begin{equation}
\label{eq1}
\frac{\left( M_{ij}c_{kj}^{-1}-M_{kj}c_{ij}^{-1}\right) ^{2}}{M_{ij}M_{kj}}=%
\frac{M_{ij}}{M_{kj}}\left( c_{kj}^{-1}-\frac{M_{kj}}{M_{ij}}%
c_{ij}^{-1}\right) ^{2}{\rightarrow }0, \: as \: T_{j}\rightarrow \infty,
\end{equation}
and on the other hand%
\begin{equation}
\label{eq2}
T_{j}.F\left( \frac{c_{1j}^{-2}}{M_{1j}},...,
\frac{c_{n_{j}j}^{-2}}{M_{n_{j}j}}\right) %
{\rightarrow }0, \: as \: T_{j}\rightarrow \infty,
\end{equation}
which achieves the proof.
\end{proof}
\section{Lower bound for the variance of the estimated reliability of {the parallel--series system}
\label{lower}}
We consider now the {parallel--series system} $S$. From expression (\ref{g}), one can write%
\[
Var\left\{ \hat{R}\right\} =R^{2} \left[ \prod\limits_{j=1}^{n}\left( \frac{%
Var\left( \hat{R}_{j}\right) }{R_{j}^{2}}+1\right) -1\right] 
\]

The following theorem gives a lower bound for the variance of $\hat{R}$.

\begin{theorem}
\label{th2}
Denote by%
\[
Q=T^{-1}R^{2} \left[
\sum\limits_{j=1}^{n}\frac{1-R_{j}}{R_{j}}\left(
\sum_{i=1}^{n_{j}}c_{ij}^{-1}\right) \right] ^{2}
\]%
then%
\[
Var\left\{ \hat{R}\right\} \geq Q
\]
\end{theorem}

\begin{proof}
Expanding the right hand side of (\ref{g}) and using (\ref{r}), one obtains%
\[
Var\left\{ \hat{R}\right\}=R^{2}%
\left[ \sum\limits_{j=1}^{n}\frac{Var\left( \hat{R}_{j}\right) }{R_{j}^{2}}%
+F\left( \frac{Var\left( \hat{R}_{1}\right) }{R_{1}^{2}},...,\frac{Var\left( 
\hat{R}_{n}\right) }{R_{n}^{2}}\right) \right],
\]

which gives with the help of Theorem \ref{th0}%
\begin{equation}
Var\left\{ \hat{R}\right\} \geq R^{2} \sum\limits_{j=1}^{n}\frac{Q_{j}}{%
R_{j}^{2}}= R^{2}%
\sum\limits_{j=1}^{n}\frac{\left( \frac{1-R_{j}}{R_{j}}
\sum\limits_{i=1}^{n_{j}}c_{ij}^{-1}\right) ^{2}}{T_{j}}
  \label{vr1}
\end{equation}%
This last expression has the form 
\[
R^{2} \sum\limits_{j=1}^{n}\frac{%
a_{j}}{T_{j}}
\]%
which can be expanded, thanks to identity (\ref{lem1}), as follows%
\begin{eqnarray}
&&R^{2} T^{-1}\left[ \sum\limits_{j=1}^{n}%
\frac{1-R_{j}}{R_{j}}\left( \sum\limits_{k=1}^{n_{j}}c_{kj}^{-1}\right) \right] ^{2}\nonumber 
 \\
&+&R^{2}
T^{-1}\sum\limits_{i=1}^{n-1}\sum\limits_{j=i+1}^{n}\frac{\left( T_{i}\frac{1-R_{j}%
}{R_{j}}\sum\limits_{k=1}^{n_{j}}c_{kj}^{-1}-T_{j}\frac{1-R_{i}}{R_{i}}%
\sum\limits_{k=1}^{n_{i}}c_{ki}^{-1}\right) ^{2}}{T_{i}T_{j}}  \label{vr2e}
\end{eqnarray}%
and as a consequence 
\[
Var\left\{ \hat{R}\right\} \geq T^{-1} R^{2} \left[ \sum_{j=1}^{n}\frac{1-R_{j}}{%
R_{j}}\left( \sum_{k=1}^{n_{j}}c_{kj}^{-1}\right) \right] ^{2}=Q,
\]%
which achieves the proof.
\end{proof}
\section{The hybrid two stage design for the {parallel--series} system $S$
\label{overlapping}}
Similarly to the case of a subsystem an from expressions (\ref{vr1}) and (%
\ref{vr2e}), one gives interest to the numbers $T_{j}$ which minimize the
quantity 
\[
\sum\limits_{i=1}^{n-1}\sum\limits_{j=i+1}^{n}\frac{\left( T_{i}\frac{1-R_{j}%
}{R_{j}}\sum\limits_{k=1}^{n_{j}}c_{kj}^{-1}-T_{j}\frac{1-R_{i}}{R_{i}}%
\sum\limits_{k=1}^{n_{i}}c_{ki}^{-1}\right) ^{2}}{T_{i}T_{j}}, 
\]%
and obtains the asymptotic optimality criteria 
\[
\frac{T_{i}}{T_{j}}=\frac{\frac{1-R_{i}}{R_{i}}\sum%
\limits_{k=1}^{n_{i}}c_{ki}^{-1}}{\frac{1-R_{j}}{R_{j}}\sum%
\limits_{k=1}^{n_{j}}c_{kj}^{-1}}, 
\]%
for all $i,j\in \left\{ 1,2,...,n\right\} $, which gives the rule%
\begin{equation}
\label{tj}
T_{j}=T\frac{\frac{1-R_{j}}{R_{j}}\sum\limits_{k=1}^{n_{j}}c_{kj}^{-1}}{%
\sum\limits_{k=1}^{n}\frac{1-R_{k}}{R_{k}}\sum%
\limits_{i=1}^{n_{k}}c_{ik}^{-1}} 
\end{equation}
We can now implement a hybrid two stage design for the determination of the
numbers $T_{j}$ as well as $M_{ij}$ as follows:

\begin{description}
\item[Stage 1] choose $L=\left[\sqrt{T}\right]$: one applies the two stage scheme given
in subsection \ref{par} for each subsystem $S_{j}$ with $T_{j}=L$ and $L_{j}=\left[\sqrt{T_{j}}\right]$. Next, obtain the predictor,
according to the rule (\ref{tj}),%
\[
\hat{T}_{j}=\left[T\frac{\frac{1-\hat{R}_{j}}{\hat{R}_{j}}\sum%
\limits_{k=1}^{n_{j}}\hat{c}_{kj}^{-1}}{\sum\limits_{k=1}^{n}\frac{1-\hat{R}%
_{k}}{\hat{R}_{k}}\sum\limits_{i=1}^{n_{k}}\hat{c}_{ik}^{-1}}\right],~j=1,\ldots ,n-1. 
\]

\item[Stage 2] define the corrector 
\begin{eqnarray*}
T_{j}&=&\max \left\{ L,\hat{T}_{j}\right\},~j=1,\ldots ,n-1, \\
T_{n}&=&T-\sum_{j=1}^{n-1}T_{j},
\end{eqnarray*}%
and take back the two stage scheme for each subsystem $S_{j}$ to calculate $%
M_{ij}$ with the sample size equals $T_{j}$.
\end{description}

Now, the main result of this paper is given by the following theorem.

\begin{theorem}
\label{th3}Choosing the $T_{j}$ and $M_{ij}$ according to the hybrid
two stage design, one obtains 
\[
\lim_{T\rightarrow \infty }T\left( Var\left\{ \hat{R}\right\} -Q\right) =0,
\]%
where $Q$ is defined in Theorem \ref{th2}.
\end{theorem}

\begin{proof}
The relation (\ref{vrj}) implies that%
\begin{eqnarray*}
& Var\left\{ \hat{R}_{j}\right\}  = Q_{j}+T_{j}^{-1}\left( 1-R_{j}\right)
^{2}\sum\limits_{i=1}^{n_{j}-1}\sum\limits_{k=i+1}^{n_{j}}\frac{\left(
M_{ij}c_{kj}^{-1}-M_{kj}c_{ij}^{-1}\right) ^{2}}{M_{ij}M_{kj}} \\
& +\left( 1-R_{j}\right) ^{2}F\left( \frac{c_{1j}^{-2}}{M_{1j}},...,\frac{c_{n_{j}j}^{-2}}{M_{n_{j}j}}\right)
\end{eqnarray*}%
As a consequence of the hybrid two stage design and the strong law of
large numbers, $T/T_{j}$ and $T_{j}/M_{ij}$ remain bounded for all $i,j$ as $%
T\rightarrow \infty $. It follows that, as $T\rightarrow \infty$,%
\[
F\left( \frac{c_{1j}^{-2}}{M_{1j}},...,\frac{%
c_{n_{j}j}^{-2}}{M_{n_{j}j}}\right) =o\left( T^{-1}\right),
\]%
and
\[
T_{j}^{-1}\sum_{i=1}^{n_{j}-1}\sum_{k=i+1}^{n_{j}}\frac{\left(
M_{ij}c_{kj}^{-1}-M_{kj}c_{ij}^{-1}\right) ^{2}}{M_{ij}M_{kj}}=o\left(
T^{-1}\right),
\]%
thanks to (\ref{eq1}) and (\ref{eq2}). Thus,%
\[
Var\left\{ \hat{R}_{j}\right\} =Q_{j}+o\left( T^{-1}\right) , \: as \:%
T\rightarrow \infty,
\]%
which implies that%
\begin{eqnarray*}
\prod\limits_{j=1}^{n}\left( \frac{Var\left\{ \hat{R}_{j}\right\} }{R_{j}^{2}%
}+1\right) &=&\prod\limits_{j=1}^{n}\left( \frac{Q_{j}}{R_{j}^{2}}+1+o\left(
T^{-1}\right) \right) \\
&=&\prod\limits_{j=1}^{n}\left( \frac{Q_{j}}{R_{j}^{2}}+1\right) +o\left(
T^{-1}\right) 
\end{eqnarray*}%

As a consequence,%
\[
\lim\limits_{T\rightarrow \infty } T\left( Var\left\{ \hat{R}\right\} -Q\right)=%
R^{2} \lim\limits_{T\rightarrow \infty }T.%
\left[ \prod\limits_{j=1}^{n}\left( \frac{Q_{j}}{R_{j}^{2}}+1\right) -1-Q%
\right]
\]%
Now, expanding the product within the limit and applying identity (\ref{lem1}),
after having replaced $Q_{j}$ by its expression (\ref{qj}), one obtains 
\begin{eqnarray*}
\prod\limits_{j=1}^{n}\left( \frac{Q_{j}}{R_{j}^{2}}+1\right) -1&=&R^{2} \left[ \sum\limits_{j=1}^{n}\frac{%
Q_{j}}{R_{j}^{2}}+F\left( \frac{Q_{1}}{R_{1}^{2}},...,\frac{Q_{n}}{R_{n}^{2}}%
\right) \right] \\
&=&Q+R^{2} \left( A+B\right), 
\end{eqnarray*}%

where	
\begin{eqnarray*}
A&=&T^{-1}\sum\limits_{i=1}^{n-1}\sum\limits_{k=i+1}^{n}\frac{\left( T_{i}\left( \frac{%
1-R_{k}}{R_{k}}\right) \left( \sum\limits_{l=1}^{n_{k}}c_{lk}^{-1}\right)
-T_{k}\left( \frac{1-R_{i}}{R_{i}}\right) \left(
\sum\limits_{l=1}^{n_{k}}c_{li}^{-1}\right) \right) ^{2}}{T_{i}T_{k}}  \\
B&=&F\left( \frac{Q_{1}}{R_{1}^{2}},...,\frac{Q_{n}}{R_{n}^{2}}\right)   
\end{eqnarray*}%

Once more, the hybrid two stage allocation scheme and the strong law of
large numbers provide%
\[
\lim_{T\rightarrow \infty }T.A=0
\]%
and%
\[
\lim_{T\rightarrow \infty }T.B=0,
\]%
which achieves the proof.
\end{proof}
\section{Monte Carlo simulation}
\label{monte carlo}
Let us remark first that the lower bound $Q$ is a first order approximation
of the optimal variance of the reliability estimate under the constraint (%
\ref{constr}) when $T$ is large.

In the first experiment, {we will validate the fact that the hybrid scheme
provides the best allocation at system level.} As in Figure \ref{fig:f1}, we consider a simple
parallel-series system of two subsystems each one, with varying
reliabilities and a fixed sample size $T=20 $. For each situation A,B,C and
D and for each partition sample size $\left\{ T_{1},T-T_{1}\right\} $ where $%
T_{1} $ varies from $\left[\sqrt{T}\right]$ to $T-\left[\sqrt{T}\right]$, {by step one }, we have applied the proper
two stage design for each parallel subsystem and
reported in a bar diagram $Var\left( \hat{R}\right) $ as a function of $T_{1}
$, see Figure \ref{fig:f2}. On the other hand, in Table \ref{tab:1}, we
have reported the expected value of $T_{1}=M_{11}+M_{21}$ given by the
hybrid two stage design. As expected, our scheme gives the best allocation
for each situation.

The second experiment deals with a non trivial
parallel-series system just as in \cite{matcom}, where subsystems are composed,
respectively, of 2,3,4 and 5 components, see Figure \ref{fig:f3}. The
partition total numbers $T_{j}$ to test in each subsystem are evaluated
systematically by the hybrid two stage design while their sum $T$ is
incremented from 100 to 10000 by step of 100. Figure \ref{fig:f4}
shows the rate of the excess of variance $T\left( Var\left( \hat{R}\right)
-Q\right) $ at logarithmic scale as a function of the sample size $T$. The
asymptotic optimality of the hybrid scheme is validated.
\begin{figure}[htbp]
\centering
\includegraphics[width=0.75\textwidth]{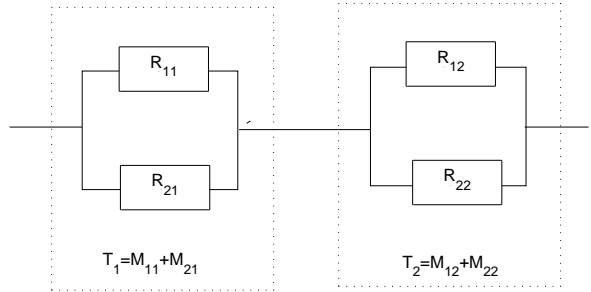}  
\caption{A simple parallel-series system of two subsystems with two components each one}
\label{fig:f1}
\end{figure}
\begin{figure}[htbp]
\centering
\includegraphics[width=0.75\textwidth]{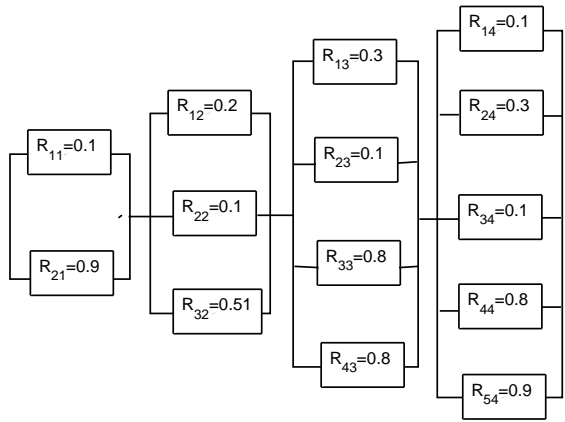}  
\caption{A non trivial parallel-series system.}
\label{fig:f3}
\end{figure}
\begin{figure}[htbp]
\centering
\includegraphics[width=0.75\textwidth]{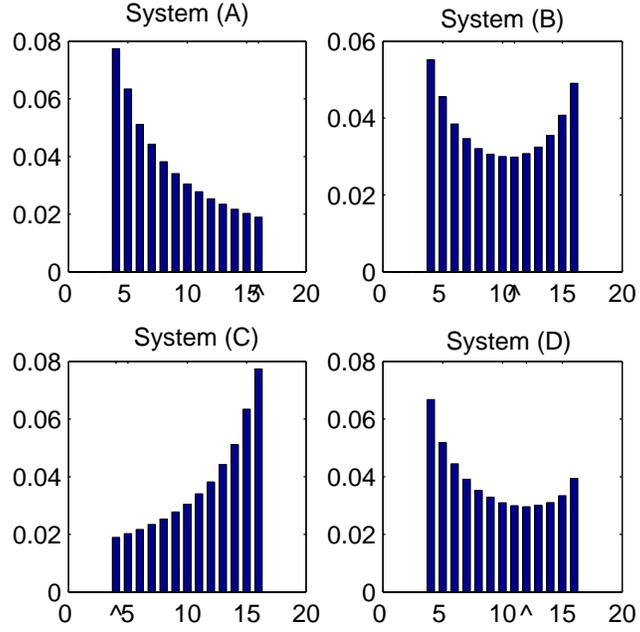}  
\caption{Bar diagram $Var\left( \hat{R}\right) $ as a function of $T_{1}$
for each case A,B,C and D : $\hat{ }$ shows the minimum of $Var\left( \hat{R}%
\right) $}
\label{fig:f2}
\end{figure}

\begin{figure}[htbp]
\centering
\includegraphics[width=0.75 \textwidth]{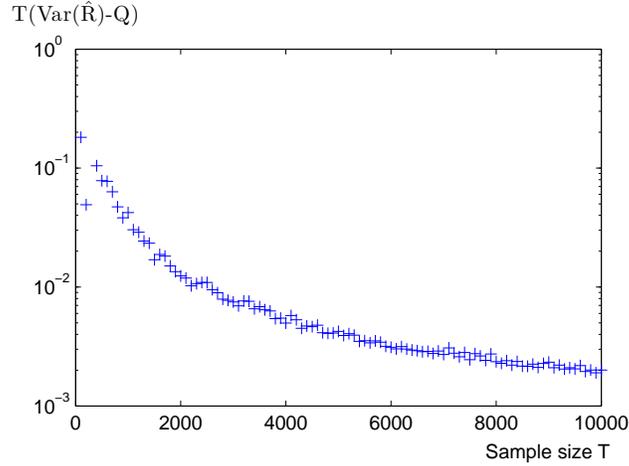}  
\caption{Asymptotic optimality of the hybrid two stage design : the speed of
the excess of variance $T\left( Var\left( \hat{R}\right) -Q\right) $ at
logarithmic scale as a function of the sample size $T$ }
\label{fig:f4}
\end{figure}

\begin{table}[htbp]{
\centering
\begin{tabular}{|l|l|l|l|l|l|}
\hline
System & $R_{11}$ & $R_{21}$ & $R_{12}$ & $R_{22}$ & $E(T_{1})$ \\ \hline
A & $0.1$ & $0.11$ & $0.9$ & $0.99$ & $16$ \\ \hline
B & $0.5$ & $0.55$ & $0.51$ & $0.6$ & $11$ \\ \hline
C & $0.9$ & $0.99$ & $0.1$ & $0.11$ & $4$ \\ \hline
D & $0.2$ & $0.4$ & $0.6$ & $0.3$ & $12$ \\ \hline
\end{tabular}
\vspace{0.5cm}
\caption{Expected value of $T_{1}=M_{11}+M_{21}$ given by the hybrid two
stage design}\label{tab:1}}
\end{table}

\section{Conclusion}
The proof of the first order asymptotic optimality for the proper two stage
design for a parallel subsystem as well as for the hybrid two stage design
for the full system has been obtained mainly through the following steps
\begin{itemize}
\item an adequate writing of the variance of the reliability estimate,
\item a lower bound for this variance, independent of allocation,
\item the allocation defined by the hybrid sampling scheme and the strong law of large numbers.
\end{itemize}
With a straightforward but tedious adaptation, the above study can be namely
extended to deal with complex systems involving a multi-criteria
optimization problem under a set of constraints such as risk, system weight,
cost, performance and others, in a fixed or in a Bayesian framework.
\section*{Acknowledgments}
This work is supported with grants by the national research project (PNR) and the L.A.A.R laboratory of the department of physics 
in the university Mohamed Boudiaf of Oran.

\end{document}